\documentclass[11pt]{amsart}
\usepackage{amsmath, color}
\usepackage{amssymb}
\usepackage{amsfonts}
\usepackage{amsthm}
\usepackage{mathrsfs}
\usepackage[all]{xy}
\usepackage{cite}
\usepackage{graphicx}
\usepackage{braket}
\usepackage{CJK}

\newcommand{\be}{\begin{equation}}
\newcommand{\ee}{\end{equation}}
\newcommand{\Tr}{\mathrm{Tr}}
\newcommand{\Con}{\mathrm{Con}}
\newtheorem{theorem}{Theorem}[section]

\theoremstyle{definition}

\theoremstyle{remark}

\numberwithin{equation}{section}

\begin{document}
\title{Quantum discord of X-states as optimization of one variable function}
\author{Naihuan Jing$^*$, Bing Yu}
\address{Jing: School of Mathematics, South China University of Technology and Department of Mathematics,
North Carolina State University}
\address{Yu: School of Mathematics, South China University of Technology}

\thanks{{\scriptsize
\hskip -0.4 true cm MSC (2010): Primary: 81P40; Secondary: 81Qxx.
$*$Corresponding author, jing@ncsu.edu}}

\maketitle

\begin{abstract}
We solve the quantum discord completely as an optimization of certain one variable function
for arbitrary two qubit X state. Exact solutions of the quantum discord
are obtained for several nontrivial regions of
the five parametric space for the quantum state. Exceptional solutions are determined via
an iterative algorithm.
\end{abstract}

\section{Introduction}

A quantum state can be studied through entanglement, separability, classical correlation and quantum correlation
\cite{OZ, HV, V, MPSVW, GBGZ, St}.
The classical and quantum correlation of a quantum state can be quantified by the notion of quantum discord.
If $\rho^{ab}$ is a bipartite quantum state, the quantum mutual information is defined by
\begin{align}
\mathcal I(\rho)=S(\rho^a)+S(\rho^b)-S(\rho^{ab})
\end{align}
where $S(\rho)=-\mathrm{Tr}\rho\log_2(\rho)$ is the von Neumann entropy of the quantum state.

To reveal the nature of quantum correlation, Olliver and Zurek \cite{OZ} proposed to use
the entropy of measurement-based conditional density operators to study the classical correlation.
The von Neumann measurement is an ensemble of projectors $B_k$ such that $\sum_kB_k=I$ where
$B_k$ are mutually orthogonal idempotents. When the measurement $\{B_k\}$ is performed locally on one party of the system $\rho^{ab}$, the quantum state $\rho$ is changed to
\begin{align}\label{e:rhok}
\rho_k=\frac1{p_k}\Tr_b(1\otimes B_k)\rho(1\otimes B_k)
\end{align}
with the probability $p_k=\Tr(1\otimes B_k)\rho(1\otimes B_k)=\Tr\rho(1\otimes B_k)$. The quantum conditional entropy is given by
\begin{align}
S(\rho|\{B_k\})=\sum_kp_kS(\rho_k).
\end{align}

The quantum mutual information with respect to $\{B_k\}$ is then defined as
\begin{align}
\mathcal I(\rho|\{B_k\})=S(\rho^a)-S(\rho|\{B_k\}),
\end{align}
and the classical correlation is measured by the quantity \cite{OZ}
\begin{align}
\mathcal C(\rho)=sup_{\{B_k\}}\mathcal I(\rho|\{B_k\}).
\end{align}
Then the quantum discord is simply defined as the difference
\begin{align}
\mathcal Q(\rho)=\mathcal I(\rho)-\mathcal C(\rho)
\end{align}

As the quantum discord is given by the supremum over the set of von Noumann measurements, mathematically
the problem is equivalent to optimization of a multi-variable function
with five parameters over a closed domain implicitly defined.
In \cite{L}, Luo found the first exact formula for the Bell diagonal state,
which corresponds to a cross-section of the general 5-dimensional region by a 3-dimensional space.
The problem for general X-states has been further studied in various works \cite{BC, WS, FWBAC, DVB}.
Notably Ali et al \cite{A} gave the ARA algorithm to calculate the quantum discord of
the X-type state by reducing the problem to an optimization of certain three variable function. Other methods
\cite{GA, Y, GGZ, GAJ, LMXW, CZYYO, Shi, LWF, WZ, H} have also been proposed to solve optimization of
multi-variable functions in different parametrization and then
claimed that in most
cases the maximum value is given by two or three
possible critical or special points. The quantum discord has also been studied using
other more general POVM measurements \cite{MHR, S, Y, YF}. So far the best record
has been some reduction to extremal problems of
two variable functions.

However,
these methods have not completely solved the problem of the quantum discord except Luo's solution
for the Bell diagonal state \cite{L} which gives an {\it exact} or {\it analytical} formula for the quantum discord.
In fact, most of the current methods are useful for many of the situations but not all cases. Their
main ideas are to
solve the optimization problem by writing down a system of partial differential equations
and then claimed that the solutions are given by those of the system.

There are three questions needed to be addressed for these methods to be successful for the quantum discord. First of all, it is impossible to
solve these systems of partial differential equations analytically. Secondly, even if one manages to solve the system of
the partial differential equations numerically, the solutions may still turn out to be some local extremal but the global ones,
since most of the currently available methods did not discuss the situation on the boundary. Therefore these differential equations
could miss the important solutions of the quantum discord. Thirdly, even the numerical solutions to those
systems of partial differential systems could be problematic, since they are usually of high dimensional, and it was observed
by Huang \cite{H} that the computational cost could grow exponentially with the dimension of the Hilbert space. The authors
failed to locate a practical numerical method in the literature to solve this problem either, as they
are optimization of multivariable functions and no satisfactory numerical methods are available for such complicated multivariable
(often ill-defined) partial differential equations.
For example, Ex. 2 in Sect. 3 cannot be solved by
any of the currently available algorithms. In fact, its solution was only obtained by
examining the graph of the quantum discord (cf. \cite{WZ}), thus its accuracy is at mercy of raw eyes.

The reason behind this problem and trouble is perhaps that the present available methods more or less use the
Lagrange multiplier, which only gives necessary conditions
for the interior critical points (see any standard calculus book).
Those exceptional solutions such as Ex. 2 in Sect. 3, on the other hand, appear not
at interior points but on the boundary. To completely understand the physical meaning of the quantum discord,
it is necessary to give a rigorous and satisfactory solution of the associated optimization problem.

The goal of this paper is to solve the quantum discord of the general X-state for all situations
on the whole domain in three steps. First, we propose
a new method to reduce
the associated optimization problem into that of a {\it one-variable} entropy-like function
on the closed interval $[0, 1]$ (cf. Theorem \ref{t:1}), which in principle solves the problem of
the quantum discord.
Second, we give exact and analytical solutions for
the general X-type state for several nontrivial regions of the parameters
and prove rigorously that the answer is mostly given by the end-points of $[0, 1]$ in Theorem \ref{t:2}.
Third, for the exceptional cases not covered by the second step and when the maximum is at an interior point of $(0, 1)$, we have formulated
an effective algorithm to pin down the exotic solutions using Newton's formula (see Theorem \ref{Th3}).
We remark that the third step covers
the situation when all previous methods cannot solve the quantum discord.
Combining with the end-points, the iterative formula has
completely resolved the
problem of the quantum discord for the general X-type state. As an example to demonstrate the power of our method, we will solve
the aforementioned Ex. 2 accurately (by six simple iterations) without resorting to its graph.

We also compare our results with the formulas of \cite{S}, where the authors have used the concurrence to compute the
quantum discord of rank two mixed states. It is verified that their formulas correspond to our special cases of
either $F(1)$ or $F(0)$.

\section{quantum discord for X states}
Let $\{\sigma_i\}_{i=1}^3$ be the standard Pauli spin matrices such that
$\{\sigma_i, \sigma_j\}=2\delta_{ij}$. It is well-known
that any two-qubit state is local unitary equivalent to the Bloch form with diagonal quadratic terms
in $\sigma_i\otimes\sigma_j$. In this paper, we restrict ourselves with the general X-type quantum state
\begin{equation}\label{Eq2.1}
\rho=\frac{1}{4}(I\otimes I+r\sigma_3\otimes I+I\otimes s\sigma_3+\sum_{i=1}^3c_i\sigma_i\otimes \sigma_i),
\end{equation}
where $r,s\in[-1,1]$ and $|c_{i}|\leqslant1$. The eigenvalues of $\rho$ are given by
\begin{gather}\label{Eq2.2}
\lambda_{1,2}=\frac{1}{4}(1-c_3\pm\sqrt{(r-s)^2+(c_1+c_2)^2}),\notag \\
\lambda_{3,4}=\frac{1}{4}(1+c_3\pm\sqrt{(r+s)^2+(c_1-c_2)^2}).
\end{gather}
Therefore the quantum state $\rho$ is defined over the following
closed region $R(\rho)$ of $\mathbb R^5$:
\begin{equation}\label{e:region}
\begin{aligned}
1-c_3&\geqslant \sqrt{(r-s)^2+(c_1+c_2)^2},\\
1+c_3&\geqslant \sqrt{(r+s)^2+(c_1-c_2)^2}.
\end{aligned}
\end{equation}
The manifold $R(\rho)$'s boundary satisfies the following obvious
constraints: $|c_3|\leqslant 1$,
$c_1^2+c_2^2\leqslant 1+c_3^2$, and $r^2+s^2\leqslant 1+c_3^2$.
$R(\rho)$ is also contained in the region bounded by the following hyperplanes:
\begin{align}
1-c_3&\geqslant |c_1+c_2|, \qquad 1-c_3\geqslant |r-s|,\\
1+c_3&\geqslant |c_1-c_2|, \qquad 1+c_3\geqslant |r+s|.
\end{align}
In particular, $1-c_3\geqslant\max\{|r|, |s|\}$ and $1-c_3\geqslant\max\{|c_1|, |c_2|\}$

To evaluate the mutual information $I(\rho)$, we need the marginal states of $\rho$:
\begin{align}\label{Eq2.3}
\rho^a&=\Tr_b\rho=diag(\frac{1}{2}(1+r), \frac12 (1-r)), \notag\\
\rho^b&=\Tr_a\rho=diag(\frac{1}{2}(1+s), \frac12 (1-s)).
\end{align}

Thus the quantum mutual information of $\rho$ is given by
\begin{equation}\label{Eq2.4}
\begin{split}
\mathcal I(\rho)=&S(\rho^a)+S(\rho^b)-S(\rho)\\
=&2-\frac{1}{2}(1+r)\log_2(1+r)-\frac{1}{2}(1-r)\log_2(1-r)\\
&-\frac{1}{2}(1+s)\log_2(1+s)-\frac{1}{2}(1-s)\log_2(1-s)+\sum_{i=1}^{4}\lambda_i\log_2\lambda_i\\
\end{split}
\end{equation}

Next we evaluate the classical correlation $\mathcal C(\rho)$. Any von Neumann measurement can be
written as $\{B_k=V|k\rangle\langle k|V^\dag: k=0,1\}$ for some $V\in \mathrm{SU}(2)$,
and each unitary matrix $V\in \mathrm{SU}(2)$ is parameterized up to a phase factor
by the 4-dimensional unit sphere $t^2+\sum_{i=1}^3y_i^2=1$
such that $V=tI+i\sum_{i=1}^3y_i\sigma_i$, $t, y_i\in\mathbb R$.

After the measurement $\{B_k\}$, the ensemble $\{\rho_k, p_k\}$ is given by \eqref{e:rhok}.

For $V=tI+i\sum_{i=1}^3y_i\sigma_i$, it follows from symmetry that
\begin{equation}\label{Eq2.7}
\begin{split}
V^\dag\sigma_1V=(t^2+y_1^2-y_2^2-y_3^2)\sigma_1+2(ty_3+y_1y_2)\sigma_2+2(-ty_2+y_1y_3)\sigma_3, \\
V^\dag\sigma_2V=(t^2+y_2^2-y_3^2-y_1^2)\sigma_2+2(ty_1+y_2y_3)\sigma_3+2(-ty_3+y_1y_2)\sigma_1, \\
V^\dag\sigma_3V=(t^2+y_3^2-y_1^2-y_2^2)\sigma_3+2(ty_2+y_1y_3)\sigma_1+2(-ty_1+y_2y_3)\sigma_2,
\end{split}
\end{equation}
and the expansion coefficients of $V^\dag\sigma_iV$ lie on the unit 3-sphere, or more precisely
the transition matrix $\in\mathrm{SO}(3)$.

Introduce new variables
$$z_1=2(-ty_2+y_1y_3), \quad z_2=2(ty_1+y_2y_3), \quad z_3=t^2+y_3^2-y_1^2-y_2^2.$$
Then $z_1^2+z_2^2+z_3^2=1$, and we have that
\begin{align}
B_0\sigma_iB_0&=z_iB_0\\
B_1\sigma_iB_1&=-z_iB_1
\end{align}
Therefore we obtain that
\begin{align}
\rho_0&=\frac{(1+sz_3)I+c_1z_1\sigma_1+c_2z_2\sigma_2+(r+c_3z_3)\sigma_3}{2(1+sz_3)}\\
\rho_1&=\frac{(1-sz_3)I-c_1z_1\sigma_1-c_2z_2\sigma_2+(r-c_3z_3)\sigma_3}{2(1-sz_3)},
\end{align}
with $p_0=\frac{1}{2}(1+sz_3), p_1=\frac{1}{2}(1-sz_3)$.
The nonzero eigenvalues of $\rho_0$ and $\rho_1$ are given by
\begin{align}\label{e:rho}
\lambda^{\pm}_{\rho_0}&=\frac{1+sz_3\pm\sqrt{r^2+2rz_3c_3+\sum_{i=1}^3(c_iz_i)^2}}{2(1+sz_3)},\\
\lambda^{\pm}_{\rho_1}&=\frac{1-sz_3\pm\sqrt{r^2-2rz_3c_3+\sum_{i=1}^3(c_iz_i)^2}}{2(1-sz_3)}.
\end{align}
From these one sees $\sqrt{r^2\pm 2rz_3c_3+\sum_{i=1}^3(c_iz_i)^2}\leqslant 1\pm sz_3$.
In particular, $\sqrt{r^2+c^2}\leqslant 1$, where $c=\max\{|c_1|, |c_2|\}$.
Following \cite{OZ, HV} the measure of the classical correlations is
\begin{align*}
\mathcal C(\rho)&=sup\{\mathcal I(\rho|{B_k})\}=sup\{S(\rho^a)-S(\rho|{B_k})\}\\
&=sup\{S(\rho^a)-(p_0S(\rho_0)+p_1S(\rho_1))\}\\
&=1-\frac{1}{2}(1+r)\log_2(1+r)-\frac{1}{2}(1-r)\log_2(1-r)\\
&+sup\{p_0 (\lambda^+_{\rho^0}\log_2\lambda^+_{\rho^0}+\lambda^-_{\rho^0}\log_2\lambda^-_{\rho^0})
+p_1 (\lambda^+_{\rho^1}\log_2\lambda^+_{\rho^1}+\lambda^-_{\rho^1}\log_2\lambda^-_{\rho^1})\},
\end{align*}
where the supremum is taken over von Neumann measures. The function inside supremum is a
function of $z_1, z_2, z_3$ subject to $z_1^2+z_2^2+z_3^2=1$, so there are two independent
variables.

Now we consider the function $G$ of two variables $z_3, \theta$:
\begin{equation}\label{Eq2.12}
\begin{split}
G(\theta, z_3)&=\frac{1}{4}(1+sz_3+\sqrt{r^2+2rc_3z_3+\theta})\log_2\frac{1+sz_3+\sqrt{r^2+2rc_3z_3+\theta}}{1+sz_3}\\
&+\frac{1}{4}(1+sz_3-\sqrt{r^2+2rc_3z_3+\theta})\log_2\frac{1+sz_3-\sqrt{r^2+2rc_3z_3+\theta}}{1+sz_3}\\
&+\frac{1}{4}(1-sz_3+\sqrt{r^2-2rc_3z_3+\theta})\log_2\frac{1-sz_3+\sqrt{r^2-2rc_3z_3+\theta}}{1-sz_3}\\
&+\frac{1}{4}(1-sz_3-\sqrt{r^2-2rc_3z_3+\theta})\log_2\frac{1-sz_3-\sqrt{r^2-2rc_3z_3+\theta}}{1-sz_3},\\
\end{split}
\end{equation}
where $\theta=\sum_{j=1}^3(c_jz_j)^2=(c_1^2-c_3^2)z_1^2+(c_2^2-c_3^2)z_2^2+c_3^2$.
Hence the quantum discord is given by
\begin{equation}\label{qdiscord}
\begin{split}
&\mathcal Q(\rho)=\mathcal I(\rho)-\mathcal C(\rho)\\
=&2-\frac{1}{2}(1+s)\log_2(1+s)-\frac{1}{2}(1-s)\log_2(1-s)+\sum_{i=1}^{4}\lambda_i\log_2\lambda_i
-\max_{z_3, \theta}G(z_3, \theta).
\end{split}
\end{equation}
Here the boundary of $(\theta, z_3)$ is determined by $z_1^2+z_2^2+z_3^2=1$.

To solve the maximum of $G(z_3, \theta)$, we first note that $G$ is an even function of $z_3$, so it is
enough to consider $z_3\in [0, 1]$. Furthermore, we can reduce the optimization to that of
a one variable function. We remark that
previously available methods have only been able to reduce the problem to that of a two-variable function.
For later purpose, we will simply write $z$ for $z_3$ from now on.

\begin{theorem}\label{t:1} Let $c=\max\{|c_1|, |c_2|\}$. Then the quantum discord of the
general $X$-state $\rho$ is given by
\begin{equation}\label{e:max}
\begin{aligned}
\mathcal Q(\rho)
=&2-\frac{1}{2}(1+s)\log_2(1+s)-\frac{1}{2}(1-s)\log_2(1-s)\\
&+\sum_{i=1}^{4}\lambda_i\log_2\lambda_i
-\max_{z\in[0,1]}F(z),
\end{aligned}
\end{equation}
where
\begin{equation}\label{e:F}
\begin{aligned}
&F(z)\\
&=\frac{1}{4}(1+sz+H_+)\log_2\frac{1+sz+H_+}{1+sz}+\frac{1}{4}(1+sz-H_+)\log_2\frac{1+sz-H_+}{1+sz}\\
&+\frac{1}{4}(1-sz+H_-)\log_2\frac{1-sz+H_-}{1-sz}+\frac{1}{4}(1-sz-H_-)\log_2\frac{1-sz-H_-}{1-sz}
\end{aligned}
\end{equation}
and $H_{\pm}=\sqrt{c^2(1-z^2)+(r\pm c_3z)^2}=\sqrt{(r^2+c^2)\pm 2rc_3z+(c_3^2-c^2)z^2}$.
\end{theorem}
\begin{proof} First we notice that $G(z,\theta)$ is a strictly increasing function of $\theta$:
 \begin{align*}
\frac{\partial G}{\partial\theta}&=\frac{1}{8\sqrt{r^2+2rc_3z+\theta}}\log_2(\frac{1+sz+\sqrt{r^2+2rc_3z+\theta}}{1+sz-\sqrt{r^2+2rc_3z+\theta}})\\
&+\frac{1}{8\sqrt{r^2-2rc_3z+\theta}}\log_2(\frac{1-sz+\sqrt{r^2-2rc_3z+\theta}}{1-sz-\sqrt{r^2-2rc_3z+\theta}})>0.
\end{align*}
Therefore there are no interior critical points and extremal points must lie on the boundary
of the domain. Since $\frac{\partial G}{\partial\theta}>0$, we further conclude that
$\max G$ takes place at the largest value of $\theta$ for some $z\in [0, 1]$.
As $z_1^2+z_2^2+z^2=1$, we have that
\begin{align*}
\theta&=(c_1^2-c_3^2)z_1^2+(c_2^2-c_3^2)z_2^2+c_3^2\\
&\leqslant (c^2-c_3^2)(z_1^2+z_2^2)+c_3^2\\
&=c^2+(c_3^2-c^2)z^2,
\end{align*}
For each fixed $z$, the maximum value $c^2+(c_3^2-c^2)z^2$ can be achieved by appropriate $z_1, z_2$.
In fact, for $c=|c_1|\geqslant |c_2|$, take $z_2=0$, then $\theta=c^2z_1^2+c_3^2z^2=c^2+(c_3^2-c^2)z^2$.
Similarly, take $z_1=0$ if $|c_1|<|c_2|$.
Therefore $\max G(z, \theta)=\max_{z\in[0,1]}G(z, c^2+(c_3^2-c^2)z^2)$, which is
explicitly given in Eqs. (\ref{e:max})-(\ref{e:F}).
\end{proof}

The optimization problem is generally nontrivial, as it has five parameters. We have the following analytic
formulas for several regions of the five parameters.
\begin{theorem}\label{t:2} For the general X-type quantum state, the quantum discord is explicitly computed according to the following cases.

(a) If  $s\geqslant0$, $rc_3\leqslant0$ and $c_3^2-c^2\geqslant src_3$ or $s=0, c_3^2\geqslant c^2$,
then the quantum discord is given by Eq. (\ref{e:max}) with
\begin{align}\label{case-a}
\begin{aligned}
\max_{z\in[0, 1]} F(z)=F(1)&=\frac{1}{4}(1+s+r+c_3)\log_2(\frac{1+s+r+c_3}{1+s})\\
&+\frac{1}{4}(1+s-r-c_3)\log_2(\frac{1+s-r-c_3}{1+s})\\
&+\frac{1}{4}(1-s+r-c_3)\log_2(\frac{1-s+r-c_3}{1-s})\\
&+\frac{1}{4}(1-s-r+c_3)\log_2(\frac{1-s-r+c_3}{1-s}).
\end{aligned}
\end{align}

(b) If $s\leqslant0$, $rc_3\geqslant0$ and $c_3^2-c^2\geqslant src_3$, then the quantum discord is given by the same formula as in (a).

(c) If $r=0$ and $c_3^2\geqslant c^2$ ($s$ is not necessarily $0$) or $r=s=0$. Let $C=\max\{|c_i|\}$, then the quantum discord is given by Eq. (\ref{e:max}) with
\begin{align}\label{case-c}
\begin{aligned}
\max_{z\in[0, 1]} F(z)&=\frac{1}{4}(1+s+C)\log_2(\frac{1+s+C}{1+s})+\frac{1}{4}(1+s-C)\log_2(\frac{1+s-C}{1+s})\\
&+\frac{1}{4}(1-s+C)\log_2(\frac{1-s+C}{1-s})+\frac{1}{4}(1-s-C)\log_2(\frac{1-s-C}{1-s}).
\end{aligned}
\end{align}
(d) If $s=rc_3\leqslant0$, $c^2=c_3^2$, and $c^2+r^2\leqslant\frac23$, then the quantum discord is given by Eq. (\ref{e:max}) with
\begin{align}\label{case-d}
\max_{z\in[0, 1]} F(z)&=\frac{1}{2}(1+\sqrt{r^2+c^2})\log_2(1+\sqrt{r^2+c^2})\notag\\
&+\frac{1}{2}(1-\sqrt{r^2+c^2})\log_2(1-\sqrt{r^2+c^2}).
\end{align}
\end{theorem}

\begin{proof} We compute the derivative of $F(z)$.
\begin{equation}\label{e:difb} 
\begin{split}
F^{\prime}(z)&=\frac{1}{4}\left(s\log_2\frac{(1+sz+H_+)(1+sz-H_+)(1-sz)^2}{(1-sz+H_-)(1-sz-H_-)(1+sz)^2}\right.\\
&\quad+\left.H_+^{\prime}\log_2\frac{1+sz+H_+}{1+sz-H_+}+H_-^{\prime}\log_2\frac{1-sz+H_-}{1-sz-H_-}\right)\\
&=\frac{1}{4}\left(s\log_2\frac{1-A_+^2}{1-A_-^2}
+\frac{rc_3+(c_3^2-c^2)z}{1+sz}\frac1{A_+}\log_2\frac{1+A_+}{1-A_+}\right.\\
&\qquad\qquad+\left.\frac{-rc_3+(c_3^2-c^2)z}{1-sz}\frac1{A_-}\log_2\frac{1+A_-}{1-A_-}\right)
\end{split}
\end{equation}
where $H_{\pm}^{\prime}(z)=H_\pm^{-1}(\pm rc_3+(c_3^2-c^2)z)$ and $\displaystyle A_{\pm}=\frac{H_\pm}{1\pm sz}\in [0, 1]$.

{\it Case (a)}. Since
\begin{equation}\label{e:AA}
\begin{aligned}
&A_+^2-A_-^2=\frac{(1-sz)^2H_+^2-(1+sz)^2H_-^2}{(1-sz)^2(1+sz)^2}\\
&=\frac{(1-sz)^2((r+c_3z)^2+c^2(1-z^2))-(1+sz)^2((r-c_3z)^2+c^2(1-z^2))}{(1-sz)^2(1+sz)^2}\\
&=\frac{(1+s^2z^2)4rc_3z-4sz[(r^2+c_3^2z^2)+c^2(1-z^2)]}{(1-sz)^2(1+sz)^2}
\end{aligned}
\end{equation}

Then the first term of $F'(z)\geqslant 0$ iff $s(A_+^2-A_-^2)\leqslant 0$, which holds if
$s\geqslant 0$ and $rc_3\leqslant 0$ or $s\leqslant 0$ and $rc_3\geqslant 0$. In particular, $r=0$
implies that $s(A_+^2-A_-^2)\leqslant 0$.

Note that $\displaystyle g(x)=\frac1x\ln\frac{1+x}{1-x}$ is a strictly increasing function on $(0, 1)$, as
\begin{align*}
g'(x)&=-\frac1{x^2}\log_2\frac{1+x}{1-x}+\frac2{x\ln 2}\frac1{1-x^2}\\
&=\frac{2}{x\ln 2}\left(\sum_{n=0}^{\infty}\frac{-x^{2n}}{2n+1}+\sum_{n=0}^{\infty}x^{2n}\right)>0.
\end{align*}
Therefore
$A_+\geqslant A_-$ iff
\begin{equation}\label{e:logA}
\frac1{A_+}\ln\frac{1+A_+}{1-A_+}\geqslant \frac1{A_-}\ln\frac{1+A_-}{1-A_-}.
\end{equation}

(i) If $s\geqslant 0,rc_3\leqslant 0$ and $c_3^2-c^2\geqslant 0$, It follows from Eq.(\ref{e:AA})
that $A_+\leqslant A_-$, then Eq. (\ref{e:logA}) implies that
\begin{align*}
F'(z)&\geqslant \frac14\left(\frac{rc_3+(c_3^2-c^2)z}{1+sz}\frac1{A_+}\log_2\frac{1+A_+}{1-A_+}+
\frac{-rc_3+(c_3^2-c^2)z}{1-sz}\frac1{A_-}\log_2\frac{1+A_-}{1-A_-}\right)\\
&\geqslant \frac14\left(\frac{rc_3+(c_3^2-c^2)z}{1+sz}\frac1{A_+}\log_2\frac{1+A_+}{1-A_+}+
\frac{-rc_3+(c_3^2-c^2)z}{1+sz}\frac1{A_-}\log_2\frac{1+A_-}{1-A_-}\right)\\
&\geqslant \frac14\left(\frac{rc_3+(c_3^2-c^2)z}{1+sz}+\frac{-rc_3+(c_3^2-c^2)z}{1+sz}\right)\frac1{A_+}\log_2\frac{1+A_+}{1-A_+}\\
&=\frac12\frac{(c_3^2-c^2)z}{1+sz}\frac1{A_+}\log_2\frac{1+A_+}{1-A_+}\geqslant 0.
\end{align*}

(ii) If $s\geqslant 0,rc_3\leqslant 0$ and $src_3\leqslant c_3^2-c^2\leqslant 0$ , we have that
\begin{align*}
F'(z)&\geqslant \frac14\left(\frac{rc_3+(c_3^2-c^2)z}{1+sz}\frac1{A_+}\log_2\frac{1+A_+}{1-A_+}+
\frac{-rc_3+(c_3^2-c^2)z}{1-sz}\frac1{A_-}\log_2\frac{1+A_-}{1-A_-}\right)\\
&\geqslant \frac14\left(\frac{rc_3+(c_3^2-c^2)z}{1+sz}+\frac{-rc_3+(c_3^2-c^2)z}{1-sz}\right)\frac1{A_-}\log_2\frac{1+A_-}{1-A_-}\\
&=\frac12\frac{(-rc_3s+c_3^2-c^2)z}{(1-sz)(1+sz)}\frac{1}{A_-}\log_2\frac{1+A_-}{1-A_-}\geqslant 0.
\end{align*}

(iii) $s=0$ and $c_3^2-c^2\geqslant 0$, so $A_{\pm}=H_{\pm}$. Note that $A_+-A_-$ has the same sign as $rc_3$ due to Eq. (\ref{e:AA}). Then
\begin{equation}\label{e:diff2}
\begin{aligned}
F^{\prime}(z)&=\frac14\left(\frac{rc_3+(c_3^2-c^2)z}{A_+}
\log_2\frac{1+A_+}{1-A_+}+\frac{-rc_3+(c_3^2-c^2)z}{A_-}\log_2\frac{1+A_-}{1-H_-}\right)\\
&=\frac{rc_3}{4}\left(\frac1{A_+}\log_2\frac{1+A_+}{1-A_+}-\frac1{A_-}\log_2\frac{1+A_-}{1-A_-}\right)\\
&\quad+\frac{(c_3^2-c^2)z}{4}\left(\frac1{A_+}\log_2\frac{1+A_+}{1-A_+}
+\frac1{A_-}\log_2\frac{1+A_-}{1-H_-}\right)\geqslant 0,
\end{aligned}
\end{equation}
We see that $F^{\prime}(z)$ is always increasing, so the maximum of $F(z)$ is $F(1)$, which is simplified
to the formula shown in Eq. (\ref{case-a}).

{\it Case (b)} is treated in two subcases.

(i) Suppose that $s\leqslant 0, rc_3\geqslant 0$, and $c_3^2-c^2\geqslant 0$. then $A_+\geqslant A_-$. Therefore
\begin{align*}
F'(z)&\geqslant \frac14\left(\frac{rc_3+(c_3^2-c^2)z}{1+sz}\frac1{A_+}\log_2\frac{1+A_+}{1-A_+}+
\frac{-rc_3+(c_3^2-c^2)z}{1-sz}\frac1{A_-}\log_2\frac{1+A_-}{1-A_-}\right)\\
&\geqslant \frac14\left(\frac{rc_3+(c_3^2-c^2)z}{1-sz}+\frac{-rc_3+(c_3^2-c^2)z}{1-sz}\right)\frac1{A_-}\log_2\frac{1+A_-}{1-A_-}\\
&=\frac12\frac{z(c_3^2-c^2)}{1-sz}\frac1{A_-}\log_2\frac{1+A_-}{1-A_-}\geqslant 0.
\end{align*}

(ii) Suppose that $s\leqslant 0,rc_3\geqslant 0$ and $src_3\leqslant c_3^2-c^2\leqslant 0$, then
\begin{align*}
F'(z)&\geqslant \frac14\left(\frac{rc_3+(c_3^2-c^2)z}{1+sz}\frac1{A_+}\log_2\frac{1+A_+}{1-A_+}+
\frac{-rc_3+(c_3^2-c^2)z}{1-sz}\frac1{A_-}\log_2\frac{1+A_-}{1-A_-}\right)\\
&\geqslant \frac14\left(\frac{rc_3+(c_3^2-c^2)z}{1+sz}+\frac{-rc_3+(c_3^2-c^2)z}{1-sz}\right)\frac1{A_+}\log_2\frac{1+A_+}{1-A_+}\\
&=\frac12\frac{(-rc_3s+c_3^2-c^2)z}{(1-sz)(1+sz)}\frac1{A_+}\log_2\frac{1+A_+}{1-A_+}
\geqslant 0.
\end{align*}
Therefore the maximum of $F(z)$ is also $F(1)$.

{\it Case (c)} is treated in two subcases: (i) $r=0$ and $c_3^2\geqslant c^2$; (ii) $r=s=0$.

First of all, the assumption of $r=0$ implies
that $H_+=H_-$ and the first term of $F'(z)$ is nonnegative. Therefore
\begin{align*}
F'(z)\geqslant H_+^{-1}(c_3^2-c^2)z\left(\log_2\frac{1+sz+H_+}{1+sz-H_+}+\log_2\frac{1-sz+H_-}{1-sz-H_-}\right)
\end{align*}
When $c_3^2\geqslant c^2$, $F'(z)\geqslant 0$ and the maximum is $F(1)$, which is given as in Eq.(\ref{case-c}).
If $s=r=0$, then
\begin{align*}
F'(z)= 2H_+^{-1}(c_3^2-c^2)z\log_2\frac{1+H_+}{1-H_+}.
\end{align*}
Therefore $\max F(z)$ is $F(1)$ or $F(0)$ according to $c_3^2\geqslant c^2$ or not. In both cases,
$\max F(z)$ is given by the same formula in Eq.(\ref{case-c}).

{\it Case (d)}. If $s=rc_3\leqslant0$, $c^2=c_3^2$, and $c^2+r^2\leqslant\frac23$, It follows from Eq.(\ref{e:AA}) that the first term $F'(z)\leqslant 0$.

Let $k(s)=\frac1{H(s)}\log_2\frac{1+sz+H(s)}{1+sz-H(s)}$, where $H(s)=\sqrt{r^2+c^2+2sz}$. Then
\begin{align}
F'(z)\leqslant&\frac{rc_3}{4H_+(z)}\log_2\frac{1+sz+H_+}{1+sz-H_+}+\frac{-rc_3}{4H_-(z)}\log_2\frac{1-sz+H_-}{1-sz-H_-}\notag\\
=&\frac{rc_3}{4}(k(s)-k(-s)).
\end{align}
As a function of $s$ we have that $H'(s)=\frac{z}{H(s)}$ and
\begin{align*}
k'(s)&=\frac{z}{H(s)}\log_2\frac{1+sz+H(s)}{1+sz-H(s)}+\frac1{H(s)\ln2}\left(\frac{z+H'(s)}{1+sz+H(s)}-\frac{z-H'(s)}{1+sz-H(s)}\right)\notag\\
&=\frac{z}{H(s)}\log_2\frac{1+sz+H(s)}{1+sz-H(s)}+\frac{2z}{H^2(s)\ln2}\frac{1+sz-H^2(s)}{(1+sz)^2-H^2(s)}\notag\\
&=\frac{z}{H(s)}\log_2\frac{1+sz+H(s)}{1+sz-H(s)}+\frac{2z}{H^2(s)\ln2}\frac{1-sz-(r^2+c^2)}{(1+sz)^2-H^2(s)}\geqslant0
\end{align*}
the inequality holds because
\begin{align*}
r^2+c^2\pm s=r^2+c_3^2\pm rc_3\leqslant r^2+c_3^2+\frac{r^2+c_3^2}{2}\leqslant1.
\end{align*}
Similarly $k'(-s)\geqslant0$, thus $\frac{rc_3}{4}(k'(s)+k'(-s))\mid_s\leqslant0$, which implies that $F'(z)\leqslant0$.

Therefore the maximum of $F(z)$ on $z\in[0,1]$ is $F(0)$, which is given by the formula in Eq.(\ref{case-d}).

\end{proof}

Remark. The theorem shows that in most cases $\max\{F(1), F(0)\}$ is the maximum. Moreover, computer-generated random sets of the 5 parameters show that the maximum is mostly given by $F(1)$. However, there are still some cases with the
optimal point $\hat{z}\neq 0, 1$ (see Example 2 below).

\section{Exceptional Solutions}
As we remarked that there are situations the optimization is attained at $\hat{z}\in (0, 1)$. To nail down this case
we give the following result.

\begin{theorem}\label{Th3} The optimization of $F(z)$ in the quantum discord has the following property.

(a) $F'(0)=F^{(3)}(0)=0$, so $F(0)$ is a critical point, and both $F(0)$ and $F(1)$ are positive;

(b)To determine if there are other optimal points, we provide the following method. We have that
\begin{align}\label{e:difb1}
F'(z)=\frac{1}{4\ln2}&\left\{s\ln\frac{((1+sz)^2-H_+^2)(1-sz)^2}{((1-sz)^2-H_-^2)(1+sz)^2}\right.\notag\\
+H_+^{\prime}&\left.\ln\frac{1+sz+H_+}{1+sz-H_+}+H_-^{\prime}\ln\frac{1-sz+H_-}{1-sz-H_-}\right\}\\
F''(z)=\frac1{2\ln 2}&\left\{\frac{(s^2+H_+'^2)(1+sz)-2sH_+H_+'}{(1+sz)^2-H_+^2}
+\frac{(s^2+H_-'^2)(1-sz)+2sH_-H_-'}{(1-sz)^2-H_-^2}\right.\notag\\
&-\left.\frac{2s^2}{1-s^2z^2}+\frac12 H_+''\ln\frac{1+sz+H_+}{1+sz-H_+}+\frac12 H_-''\ln\frac{1-sz+H_-}{1-sz-H_-}\right\},
\end{align}
where $H_{\pm}=\sqrt{(r\pm c_3z)^2+c^2(1-z^2)}$, $H'_{\pm}=H_{\pm}^{-1}(\pm rc_3+(c_3^2-c^2)z)$ and $H''_{\pm}=H_{\pm}^{-3}c^2(c_3^2-c^2-r^2)$.
If the following Newton's iterative formula gives points inside $(0, 1)$, then its limit will be another optimal
point.
\begin{equation}\label{newton}
z_{k+1}=z_k-\frac{F'(z_k)}{F''(z_k)}.
\end{equation}
In practice, one usually starts with $z_0=1$.
\end{theorem}
\begin{proof} Note that $F(z)=f(z)+f(-z)$ for some analytic function $f(z)$, therefore $F^{(n)}(0)=0$
for any odd integer $n\in\mathbb N$. The iterative formula follows from Newton's approximation formula.
\end{proof}
\medskip

Example 1. Let $\rho=\frac14(I+\sum_{i=1}^3 c_i\sigma_i\otimes \sigma_i)$ be the Bell-diagonal state. Then $r=s=0$. This is a special
case of Theorem \ref{t:2} (c), so the maximum of $F(z)$ on $z\in[0,1]$ is
$\frac{1}{2}(1+C)\log_2(1+C)+\frac{1}{2}(1-C)\log_2(1-C)$, and the quantum discord is
\begin{equation}\label{Eq2.28}
\begin{split}
\mathcal Q(\rho)&=\frac{1}{4}(1-c_3+c_1+c_2)\log_2(1-c_3+c_1+c_2)\\
&+\frac{1}{4}(1-c_3-c_1-c_2)\log_2(1-c_3-c_1-c_2)\\
&+\frac{1}{4}(1+c_3+c_1-c_2)\log_2(1+c_3+c_1-c_2)\\
&+\frac{1}{4}(1+c_3-c_1+c_2|)\log_2(1+c_3-c_1+c_2)\\
&-\frac{1}{2}(1+C)\log_2(1+C)-\frac{1}{2}(1-C)\log_2(1-C)\\
\end{split}
\end{equation}
which was first given in \cite{L} where $C=\max\{|c_i|\}$.
Note that the Werner state $
\rho=a|\psi^-\rangle\langle\psi^-|+\frac{1-a}{4}I$,
where $0\leqslant a\leqslant1$, is a special case with $r=s=0,c_3=-a,c_1=c_2=-a$.
\medskip

Example 2. Let $\rho$ be the following density matrix, which
is an example that cannot be treated by previous algorithms (cf. \cite{LMXW}). 
\begin{equation}\label{counterex}
\begin{pmatrix}
       0.0783&0&0&0\\
       0&0.1250&0.1000&0\\
       0&0.1000&0.1250&0\\
       0&0&0&0.6717\\
\end{pmatrix}.
\end{equation}
The eigenvalues of $\rho$ are $\lambda_1=0.025,\lambda_2=0.0783,\lambda_3=0.2250,\lambda_4=0.6717$. In terms of the Bloch form, $r=s=-0.5934,c_3=0.5,c_1=c_2=0.2$, so $c=0.2$.

Although $F(1)>F(0)$, $F(1)$ is not a maximal value. We can solve the exceptional solution easily by Eq. (\ref{newton}).
In fact, starting with $z_0=1$,
Eq. (\ref{newton}) gives that $z_1=0.9205, z_2=0.8884, z_3=0.8833, z_4=0.8831, z_5=0.88313, z_6=0.883131.$
Therefore $\hat{z}=0.88313$ is another critical point of $F(z)$ and the maximum value.
Finally the quantum discord turns out to be
$\mathcal Q(\rho)=2-\frac{1}{2}(1+s)\log_2(1+s)-\frac{1}{2}(1-s)\log_2(1-s)+\sum_{i=1}^{4}\lambda_i\log_2\lambda_i-
F(0.88313)=0.1328$.
\medskip

 Example 3. Let
 $\rho=\frac13\{(1-a)|00\rangle\langle00|+ 2|\psi^+\rangle\langle\psi^+|+a|11\rangle\langle11|\}$, where
 $|\psi^+\rangle=\frac1{\sqrt2}(|01\rangle+|10\rangle)$, $0\leq a\leq1$
considered in \cite{A}. The eigenvalues of $\rho$ are $\lambda_1=0,\lambda_2=\frac23,\lambda_3=\frac{1-a}{3},\lambda_4=\frac{a}{3}$.
 Here $r=\frac13-\frac23a,s=\frac13-\frac23a,c_3=-\frac13,c_1=c_2=\frac23,c=max\{|c_1|,|c_2|\}=\frac23$. It
 can be checked that $F(0)$ is the maximal value using $F'(z)$. 
 The behavior of $F'(z)$ is depicted in Fig. 1 as a function of $z$ and $a$.  The quantum discord is $\mathcal Q(\rho)=2-\frac{1}{2}(1+s)\log_2(1+s)-\frac{1}{2}(1-s)\log_2(1-s)+\sum_{i=1}^{4}\lambda_i\log_2\lambda_i-F(0)$.

\begin{figure}[!h]
  \centering
  \includegraphics[width=3.5in]{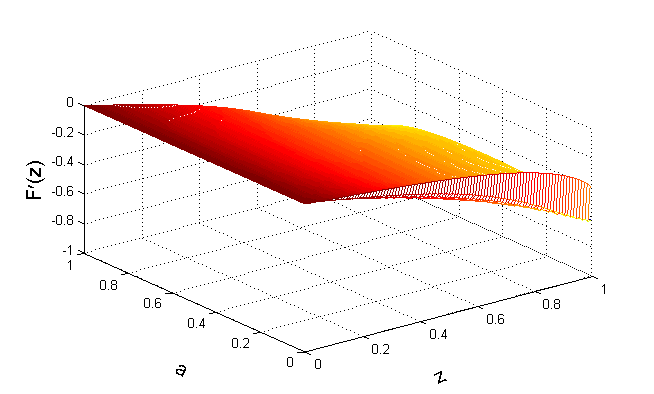}\\
  \caption{$F'(z)\leqslant 0$ in Ex. 3 and $F'(0)=0$ for all $a$.}\label{1}
\end{figure}

\section{Relationship with Concurrence}

Concurrence is an important measurement of entanglement. Its computation is a highly nontrivial problem for bipartite states.
In \cite{W1} the concurrence of the general 2-qubit $\rho$ are given in terms of the eigenvalues of $\rho$ and
an associated state $\tilde{\rho}$. In \cite{S}, the authors have computed the quantum discord of general rank two 2-qubit
in terms of the entanglement of formation using the Koashi-Winter relation.
In this section, we will show that their formula at the case of an X-state
is a special case of our general result in the case of rank two X-states.

Recall that the entanglement of formation (EoF) of a mixed state $\rho$ of two qubit is given by
\begin{equation}\label{eof}
E(\rho)=H(\frac{1+\sqrt{1-[\Con(\rho)]^2}}{2})
\end{equation}
where the function $H(x)$ is defined as
\begin{equation}
H(x)=-x\log_2x-(1-x)\log_2(1-x)
\end{equation}
for $x\in[0,1]$.
The concurrence of $\rho$ is computed by \cite{W1}
\begin{equation}
\Con(\rho)=\max\{0,\mu_1-\mu_2-\mu_3-\mu_4\}
\end{equation}
where the $\mu_i$'s in decreasing order are the square roots of the eigenvalues of the non-Hermitian matrix
$\rho\tilde{\rho}$, where $\tilde{\rho}=(\sigma_y\otimes\sigma_y)\rho^*(\sigma_y\otimes\sigma_y)$
and $\sigma_y$ is the Pauli spin matrix.

Let $\rho=\rho^{abc}$ be a tripartite state on the Hilbert space $H_a\otimes H_b\otimes H_c$.
Then two states $\rho_1$ and $\rho_2$ on $H_a\otimes H_b$ and $H_b\otimes H_c$ respectively
are called {\it complementary} if $\rho_1=\rho^{ab}=\Tr_c\rho$ and $\rho_2=\rho^{bc}=\Tr_a\rho$.
Then the Koashi-Winter relation \cite{KW} states that
\begin{equation}
C(\rho^{ab})+E(\rho^{bc})=S(\rho^b),
\end{equation}
from which one can express the quantum discord in terms of the concurrence using \eqref{eof}.

Recall that $\rho^{ab}$ is a mixed state of X-type:
\begin{equation}\label{rho1}
\rho=\frac14\left(
            \begin{array}{cccc}
              1+r+s+c_3 & 0 & 0 & c_1-c_2 \\
              0 & 1+r-s-c_3 & c_1+c_2 & 0 \\
              0 & c_1+c_2 & 1-r+s-c_3 & 0 \\
              c_1-c_2 & 0 & 0 & 1-r-s+c_3 \\
            \end{array}
          \right).
\end{equation}
Suppose that $\rho$ is of rank two. As the eigenvalues of
$\rho^{ab}$ are given in Eq. (\ref{Eq2.2}), there are only three possibilities:

(i) $\lambda_{1,2}=0, \lambda_{3,4}\neq0$. Then $c_3=1, r=s, c_1=-c_2$

(ii) $\lambda_{1,2}\neq0, \lambda_{3,4}=0.$ Then $c_3=-1, r=-s, c_1=c_2$.

(iii) $\lambda_{1,3}\neq0, \lambda_{2,4}=0$ (or $\lambda_{2,4}\neq0, \lambda_{1,3}=0$ which is similar).
In this case we have that
\begin{align*}
\sqrt{(r-s)^2+(c_1+c_2)^2}&=1-c_3,\\
\sqrt{(r+s)^2+(c_1-c_2)^2}&=1+c_3,
\end{align*}
subsequently
\begin{equation*}
\lambda_{1}=\frac{1}{2}(1-c_3),\lambda_{3}=\frac{1}{2}(1+c_3).
\end{equation*}

Cases (i) and (ii) can be treated similarly, so we look at case (i) closely.
Let's write the eigenstates of $\lambda_{3, 4}$ as
\begin{equation}
\ket{\varphi_3}=a_0\ket{0}\ket{0}+b_0\ket{1}\ket{1},
\ket{\varphi_4}=a_1\ket{1}\ket{0}+b_1\ket{1}\ket{1},
\end{equation}
where $|a_k|^2+|b_k|^2=1$ for $k=0,1$.
It follows that
\begin{gather*}
a_0=\pm\frac{r-\sqrt{r^2+c_1^2}}{\sqrt{{2(r^2+c_1^2)-2r\sqrt{r^2+c_1^2}}}},
b_0=\frac{|c_1|}{\sqrt{2(r^2+c_1^2)-2r\sqrt{r^2+c_1^2}}},\\
a_1=\pm\frac{r+\sqrt{r^2+c_1^2}}{\sqrt{{2(r^2+c_1^2)+2r\sqrt{r^2+c_1^2}}}},
b_1=\frac{|c_1|}{\sqrt{2(r^2+c_1^2)+2r\sqrt{r^2+c_1^2}}},
\end{gather*}
where $\pm$ is the sign of $c_1$.

Attaching a qubit on $H_c$ to the qubits $a$ and $b$, we obtain the purification of $\rho^{ab}$ as
\begin{equation}
\ket{\Psi}=\sqrt{\lambda_3}\ket{\varphi_3}\otimes\ket{0}+\sqrt{\lambda_4}\ket{\varphi_4}\otimes\ket{1}.
\end{equation}
Then the reduced state $\rho^{bc}=\Tr_A\ket{\Psi}\bra{\Psi}$ is
\begin{equation}
\rho^{bc}=\left(
            \begin{array}{cccc}
              \lambda_3a_0^2 & \sqrt{\lambda_3\lambda_4}a_0a_1 & 0 & 0 \\
              \sqrt{\lambda_3\lambda_4}a_0a_1 & \lambda_4a_1^2 & 0 & 0 \\
              0 & 0 & \lambda_3b_0^2 & \sqrt{\lambda_3\lambda_4}b_0b_1 \\
              0 & 0 & \sqrt{\lambda_3\lambda_4}b_0b_1 & \lambda_4b_1^2 \\
            \end{array}
          \right)
\end{equation}

The only nonzero $\mu_i$'s of $\rho^{bc}\tilde{\rho}^{bc}$ are $\sqrt{(b_0a_1-a_0b_1)^2\lambda_3\lambda_4}$
with multiplicity two. Thus $\Con(\rho^{bc})=0$, therefore $E(\rho^{bc})=0$ by Eq.({\ref{eof}}).
Subsequently their result of the quantum discord matches with our formula given by $F(0)$,
as $E(\rho^{bc})=1-\max F(z)$. Here one can see that the maximum is attained at $z=0$ by our formulas for
$F'(z)$ and $F''(z)$.

Next we look at case (iii). Write the eigenvectors of $\lambda_{1, 3}=\frac{1}{2}(1\mp c_3)$ as
\begin{equation}
\ket{\varphi_1}=a_0\ket{0}\ket{0}+b_0\ket{1}\ket{1},
\ket{\varphi_3}=a_1\ket{1}\ket{0}-b_1\ket{0}\ket{1},
\end{equation}
and $|a_k|^2+|b_k|^2=1$ for $k=0,1$. Then $\rho^{ab}$ is given by
\begin{equation}\label{rho2}
\rho^{ab}=
\left(
  \begin{array}{cccc}
    \omega_0a_0^2 & 0 & 0 & \omega_0a_0b_0 \\
    0 & \omega_1b_1^2 & -\omega_1a_1b_1 & 0\\
    0 & -\omega_1a_1b_1 &\omega_1a_1^2 & 0 \\
    \omega_0a_0b_0 & 0 & 0 & \omega_0b_0^2 \\
  \end{array}
\right)
\end{equation}
After attaching a qubit $c$ to the qubits $a$ and $b$, the purification of $\rho^{ab}$ is given by
\begin{equation}
\ket{\Psi}=\sqrt{\omega_0}\ket{\varphi_0}\otimes\ket{0}+\sqrt{\omega_1}\ket{\varphi_1}\otimes\ket{1},
\end{equation}
and the reduced state $\rho^{bc}$ is
\begin{equation}
\rho^{bc}=
\left(
  \begin{array}{cccc}
    \omega_0a_0^2 & 0 & 0 & -\sqrt{\omega_0\omega_1}a_0b_1 \\
    0 & \omega_1a_1^2 & \sqrt{\omega_0\omega_1}a_1b_0 & 0\\
    0 & \sqrt{\omega_0\omega_1}a_1b_0 & \omega_0b_0^2 & 0 \\
    -\sqrt{\omega_0\omega_1}a_0b_1 & 0 & 0 & \omega_1b_1^2 \\
  \end{array}
\right).
\end{equation}

The $\mu_i$'s of $\rho^{bc}\tilde{\rho}^{bc}$  are $\{2\sqrt{\omega_0\omega_1}a_0b_1, 2\sqrt{\omega_0\omega_1}a_1b_0, 0, 0\}$.
Thus the concurrence of $\rho{bc}$ is
\begin{equation}
\Con(\rho{bc})=2\sqrt{\omega_0\omega_1}|a_0b_1-b_0a_1|.
\end{equation}

From \eqref{rho2} and \eqref{rho1}, we have the following relations:
\begin{align*}
\omega_0a_0^2&=\frac14(1+r+s+c_3),\quad \omega_1b_1^2=\frac14(1+r-s-c_3),\\
\omega_0a_1^2&=\frac14(1-r+s-c_3),\quad\omega_1b_0^2=\frac14(1-r-s+c_3),\\
\omega_0a_0b_0&=\frac14(c_1-c_2),\quad \omega_1a_1b_1=-\frac14(c_1+c_2).
\end{align*}

Then $\Con^2(\rho^{bc})=\frac12(1+r^2-s^2-c_3^2-c_1^2+c_2^2)$, subsequently
\begin{align}\notag
E(\rho^{bc})=&-\frac{1+\sqrt{\frac12(1-r^2+s^2+c_3^2+c_1^2-c_2^2)}}{2}\log_2\frac{1+\sqrt{\frac12(1-r^2+s^2+c_3^2+c_1^2-c_2^2)}}{2}\\ \notag
&-\frac{1-\sqrt{\frac12(1-r^2+s^2+c_3^2+c_1^2-c_2^2)}}{2}\log_2\frac{1-\sqrt{\frac12(1-r^2+s^2+c_3^2+c_1^2-c_2^2)}}{2}\\
=&-\frac{1+\sqrt{s^2+c_1^2}}{2}\log_2\frac{1+\sqrt{s^2+c_1^2}}{2}-\frac{1-\sqrt{s^2+c_1^2}}{2}\log_2\frac{1-\sqrt{s^2+c_1^2}}{2}.
\end{align}

This matches again with our result that $1-\max F(z)=1-F(1)$ in the case, where it can be seen that
the maximum is indeed attained at $z=1$. Note that we need to switch $r$ and $s$ in our formula
as their paper conducts the measurement on particle $a$.

\section{Conclusions}

The quantum discord is one of the important quantum correlations, but hard to compute analytically as
it is an optimization problem over the set of von Neumann measures. For the general X-type quantum state,
we have reduced the optimization to that of a one variable
function $F(z)$ on $[0, 1]$. Several exact formulas
are given for various regions of the quantum state and an effective iterative algorithm is provided to
find the quantum discord in all situations.

Our results show that the quantum discord is usually given by $\max\{F(0), F(1)\}$, which agrees with previous
algorithms given by \cite{A, LMXW, CZYYO, LWF, Shi} (see also \cite{H}). As an example,
we check that the quantum discord of any rank two mixed state of $X$-type are
always given by either $F(1)$ or $F(0)$ as well as their relations with
the Koashi-Winter relation and entanglement of formation in the last section, where
we show that the results of \cite{S} matches exactly with the special cases of our result.
We remark that Theorem 2 is the first general result that establishes
rigorously the quantum discord is mostly given by
either $F(0)$ or $F(1)$.

As pointed out in \cite{LMXW} there are
counterexamples to many of the existing algorithms. To address this problem,
we have carefully discussed the optimization
on the boundary of the domain and found an iterative formula to compute exactly
the other possible optimal solutions, therefore completely solved the problem of the quantum discord
for the general X-type state. Using an example, our method is demonstrated to be able to
treat some quantum discord problem that cannot be solved analytically by other methods.

In summary, the problem of quantum discord for the general X-type two-qubits
(with respect to the von Neumann measurements) is completely settled in our paper.

\bigskip

\bigskip

\centerline{\bf Acknowledgments}
The first author thanks Shao-Ming Fei for helpful discussions on concurrence and
related problems.
The work is partially supported by
Simons Foundation grant No. 198129 and NSFC grant Nos. 11271138 and 11531004.

\bibliographystyle{amsalpha}

\begin{thebibliography}{99}
\bibitem{OZ} H. Ollivier, W.~H. Zurek, Quantum discord: a measure of the quantumness of correlations,
Phys. Rev. Lett. 88, 017901 (2001).
\bibitem{HV} L. Henderson, V. Vedral, Classical, quantum and total correlations, J. Phys. A 34, 6899 (2001).
\bibitem{V} V. Vedral, Classical correlations and entanglement in quantum measurements,
Phys. Rev. Lett. 90, 050401 (2003).
\bibitem{MPSVW} K. Modi, T Paterek, W. Son, V. Vedral, M. Williamson, Unified view of quantum and classical
correlations. Phys. Rev. Lett. 104, 080501 (2010).
\bibitem{St} A. Streltsov, Quantum discord and its role in quantum information theory, In Quantum Correlations beyond Entanglement, pp. 22-43.
Springer Briefs in Phys., Springer, New York, 2015. 
\bibitem{GBGZ} G.~L. Giorgi, B. Bellomo, F. Galve, R. Zambrini, Genuine quantum and classical correlations in
multipartite systems. Phys. Rev. Lett. 107, 190501 (2011).
\bibitem{L} S. Luo, Quantum discord for two-qubit systems,
Phys. Rev. A 77, 042303 (2008).
\bibitem{BC} B. Bylivka, D. Chru\'{s}ci\'{n}ski, Quantum discord between two moving two-level atoms,
Phys. Rev. A 81, 062102 (2010).
\bibitem{WS} T. Werlang, S. Souza, F.~M. Cucchietti, A.~O. Caldeira, Robustness of quantum discord to sudden death,
Phys. Rev. A 80, 024103 (2009).
\bibitem{FWBAC} F.~F. Fanchini, T. Werlang, C.~A. Brasil, L.~G.~E. Arruda, A.~O. Caldeira, Non-Markovian
dynamics of quantum discord, Phys. Rev. A 81, 052107 (2010).
\bibitem{DVB} B. Dak\'{\i}c, V. Vedral, \v{C}. Bruckner, Necessary and sufficient condition for nonzero quantum discord, Phys. Rev. Lett. 105, 190502 (2010).
\bibitem{A} M. Ali, A.~R.~P. Rau, G. Alber, Quantum discord for two-qubit X states, Phys. Rev. A 81, 042105 (2010).
\bibitem{GA} D. Girolami, G. Adesso, Quantum discord for general two--qubit states: analytical progress,
Phys. Rev. A 83, 052108 (2011). 
\bibitem{LMXW} X.-M. Lu, J. Ma, Z. Xi, X. Wang,
Optimal measurements to access classical correlations of two-qubit states, Phys. Rev. A 83, 012327 (2011).
\bibitem{GGZ} F. Galve, G. Giorgi, R. Zambrini, Orthogonal measurements are almost sufficient for quantum discord of two qubits,
Europhys. Lett. 96, 40005 (2011). 
\bibitem{GAJ} N. Quesada, A. Ali-Qasimi, D.~F.~V. James, Quantum properties and dynamics of X states,
J. Mod. Opt. 59, 1322 (2012). 
\bibitem{CZYYO} Q. Chen, C. Zhang, S. Yu, X.-X. Yi, C.-H. Oh, Quantum discord of two-qubit X states. Phys. Rev. A
84, 042313 (2011).
\bibitem{LWF} B. Li, Z.-X. Wang, S.-M. Fei, Quantum discord and geometry for a class of two-qubit states,
Phys. Rev. A 83, 022321 (2011).
\bibitem{Shi} M. Shi, C. Sun, F. Jiang, X. Yan, J. Du, Optimal measurement for quantum discord of two-qubit
states. Phys. Rev. A 85, 064104 (2012).
\bibitem{Y} M. Yurischev, On the quantum discord of general X states, Quant. Inf. Process 14, 3399 (2015).
\bibitem{WZ} X. Wu, T. Zhou, Quantum discord for the general two-qubit case,
Quant. Inf. Process 14, 1959 (2015).
\bibitem{H} Y. Huang,
Computing quantum discord is NP-complete, New J. Phys. 16, 033027 (2014).
\bibitem{S} M. Shi, W. Yang, F. Jiang, J. Du,
Quantum discord of two-qubit rank 2 states, J. Phys. A: Math. Theor. 44, 415304 (2011).
\bibitem{MHR} A. Maldonaldo-Trapp, A. Hu, L. Roa, Analytical solutions and criteria for the quantum discord of
two-qubit X-states, Quant. Inf. Process 14, 1947 (2015).
\bibitem{YF} B. Ye, S.-M. Fei, A note on one-way quantum deficit and quantum discord, Quant. Inf. Process 15, 279 (2016).
\bibitem{W1} W. K. Wooters, Entanglement of formation of an arbitrary state of two qubits, Phys. Rev. Lett. 80, 2245 (1998).
\bibitem{KW} M. Koashi, A. Winter, Monogamy of quantum entanglement and other correlations, Phys. Rev. A 69, 022309 (2004).
\end{thebibliography}

\end{document}